\documentclass[article]{llncs}

\usepackage{vdm-rg}
\usepackage{graphicx}

\usepackage[bitstream-charter]{mathdesign}
\usepackage[T1]{fontenc}

\usepackage{bbding} 

\def\InfRule#1#2#3{\[\raisebox{-0.5\baselineskip}{\mbox{\hbox{\raisebox{0.5\baselineskip}{\framebox{$#1$}}}$\begin{array}[b]{l}#2\\ \hline #3\end{array}$}}\]}

\begin{document}

\title{Handling Complex Interference Compositionally: \\ 
A Concurrent Garbage Collector Case Study}

\author{Nisansala P. Yatapanage \Envelope} 
\institute{School of Computing, The Australian National University \email{nisansala.yatapanage@anu.edu.au}}

%\correspond{nisansala.yatapanage@anu.edu.au}

\maketitle

%comment out to restore standard runningheads
\newcommand{\version}{??}
\newcommand{\kopf}{\textnormal{Handling Complex Interference Compositionally}}
\pagestyle{myheadings}
\markboth{\kopf}{\kopf}
%
%\fbox{Draft: \version}\hfill
%\fbox{Dated: \today}

%%%%%%%%%%%%%%%%%%%%%%%%%%%%%%%%%%%%%%%%%%%%%%%%
\begin{abstract}
Rely/Guarantee is a well-known verification technique for reasoning about concurrent programs. However, for some algorithms, devising suitable rely and guarantee conditions is challenging, due to the strong interference exhibited in these algorithms. The Ben-Ari concurrent garbage collector is an algorithm where the complex interactions between the components prevent the construction of compositional rely and guarantee conditions. This paper investigates an approach for verifying the Ben-Ari algorithm, which enables reasoning to be performed in a more compositional manner in cases where compositional reasoning would not otherwise be possible. This is accomplished by reasoning that a given property holds for all instances of a particular variable and then instantiating the variable to the local variable required. As well as providing a reasoning approach which is more compositional, the result is the identification of the core property required of a component, thus enabling a deeper understanding about the reasons why the algorithm works correctly. 
\end{abstract}

% !TEX root = gc-new.tex
%
% \newpage

%%%%%%%%%%%%%%%%%%%%%%%%%%%%%%%%%%%%%%%%%%%%%%%%
%J%
\section{Introduction}

Interference is a key aspect of concurrent algorithms, yet the study of interference has received very little attention. Concurrency creates complex interactions which make verification challenging. The core challenge is in specifying the interference that concurrent processes exhibit on each other.

Rely/guarantee \cite{Jones83a,Jones83b} is a well-known reasoning technique revolving around interference. The interference a process can accept from its environment is specified abstractly, as the $rely$ predicate, without needing to refer to the detailed behaviour of the environment processes. The $guarantee$ of a process is a predicate describing the behaviour that the process exhibits. To verify that two processes can operate together concurrently, it must be shown that the guarantee of one process implies the rely of the other, and vice versa. Rely/guarantee is a useful technique which works well for many concurrent algorithms. However, there is a class of algorithms for which rely/guarantee properties cannot easily be devised. Understanding the reasons why it is difficult to devise suitable rely/guarantee properties is essential, in order to develop suitable verification techniques for these algorithms. 

The algorithms which have been successfully verified using rely/guarantee fall into certain classes. The first group are the classic rely/guarantee examples, often used by Cliff Jones, the creator of rely/guarantee. These problems include \textit{findp}, which uses parallel threads to speed up the search for the first index of an array satisfying a given property $p$, and the parallel Sieve of Eratosthenes, which uses parallel threads to identify the set of prime numbers below a given maximum value. See \cite{HayesJones18} for rely/guarantee developments of these problems. These problems all fall into the same class, which contain algorithms that are inherently parallel rather than concurrent. While there is some interference when writing the results of the parallel searches, this is only minor interference handled using constructs such as Compare-and-Swap. 

Rely/guarantee is often combined with other logics, in particular separation logics, as with RGSep \cite{marriage-RGSep}. RGSep uses rely/guarantee for the shared parts of the memory and separation logic for the local parts. These logics are ideal for problems which have an aspect of separation present.

Non-blocking concurrent algorithms are the most challenging type to verify. Rely/ guarantee has been used for verifying a non-blocking concurrent garbage collector \cite{JonesYatapanage19}, Simpson's four-slot algorithm for concurrent readers and writers \cite{PVEaCfC}, and concurrent stack and queue algorithms in \cite{Yatapanage25}. For all of these approaches, it was concluded that rely and guarantee conditions could not be devised in a straightforward manner. In \cite{PVEaCfC}, in order to verify Simpson's four-slot algorithm, a new notion, \textit{possible values} was devised, to essentially describe the history of values held by a variable throughout the program execution. In \cite{JonesYatapanage19}, the notion of possible values was again used, in order to handle the unpredictable nature of the problem.

In \cite{Yatapanage25}, it is demonstrated that some concurrent stacks and queues pose problems for constructing rely conditions, because multiple environment processes could interfere with each other, thus making it difficult to identify conditions which hold over any arbitrary set of environment steps.

In \cite{JonesYatapanage19}, Ben-Ari's concurrent garbage collector is specified with rely/guarantee, but it is concluded that there is no way of specifying suitable rely and guarantee conditions without breaking compositionality. The Ben-Ari algorithm allows both the $Mutator$ and the $Collector$ to run at the same time without corrupting the memory. The difficulty with devising suitable rely and guarantee conditions is due to a specific interaction between the components, where the interference could cause memory in use to be incorrectly reclaimed by the $Collector$. While the algorithm prevents this, the reason why it works is not obviously clear.

This paper investigates the interference in this algorithm further. In particular, the interference is handled using a novel proof technique, whereby the components are verified without reasoning using the other component's local variables. Instead, the rely of the $Collector$, which contains a local variable, is shown to be satisfied by the $Mutator$ by instead proving that the $Mutator$ maintains a certain property over all the variables of a set. When instantiated with the specific variable required by the $Collector$, the rely is shown to hold. The property essentially acts as a bridge between the guarantee of the $Mutator$ and the rely of the $Collector$. This thereby allows the developers of both processes to ensure that the component interactions are safe, without having to specifically reference local variables of one process in the other's guarantee. This idea is briefly mentioned in \cite{JonesYatapanage19} as one possible verification method, but not explored in detail. 

There are two interesting aspects which arise from this. The first is the proof technique itself, which is applicable to other problems. The technique allows strongly intertwined behaviour to be separated, without requiring ghost variables. The second interesting aspect is the abstract predicate which is used for the algorithm. The predicate provides more than a simple bridge between the two components; it uncovers the core property that is required by the $Collector$. This reveals a deeper insight into \textit{why} the algorithm works. Many verification approaches focus on showing that an algorithm is correct, but do not provide any understanding of why they work. Understanding this can aid in devising further concurrent algorithms.

Finally, the approach also provides insights into the nature of interference and why rely/guarantee is insufficient for some algorithms. The rest of the paper is structured as follows. Section \ref{rgsection} provides a brief overview of rely/guarantee reasoning, Section \ref{holdsforall} introduces the proposed approach, \textit{Holds For All}, Section \ref{gcsection} describes the Ben-Ari garbage collector case study, Section \ref{AbstractVariable} explores the properties required by the $Collector$ and Section \ref{holdsreasoning} shows the reasoning process for this problem using \textit{Holds For All}. Section \ref{Discussion} gives a discussion, including comparing related work, and Section \ref{Conclusion} is the conclusion.

\section{Rely/Guarantee Reasoning} \label{rgsection}

Rely/guarantee reasoning \cite{Jones83a,Jones83b} is a compositional reasoning approach based on Hoare logic. Instead of the usual Hoare triples, there are quintuples, with an additional \textit{rely} and \textit{guarantee} condition. For each process, its rely condition specifies what it is able to handle from the environment, while the guarantee details what the process ensures to do. The rely/guarantee laws (see \cite{HayesJones18} for an introduction using an algebraic approach) require that the guarantee of each process must satisfy the rely of the other, ensuring that the processes can run concurrently. Figure \ref{rg} provides an overview. Note that the rely and guarantee conditions must hold transitively over all of the environment, respectively program, steps.

   \begin{figure}
   \begin{center}
    \includegraphics[alt={A timeline showing how rely, guarantee, pre and post conditions relate to each other},width=0.8\linewidth]{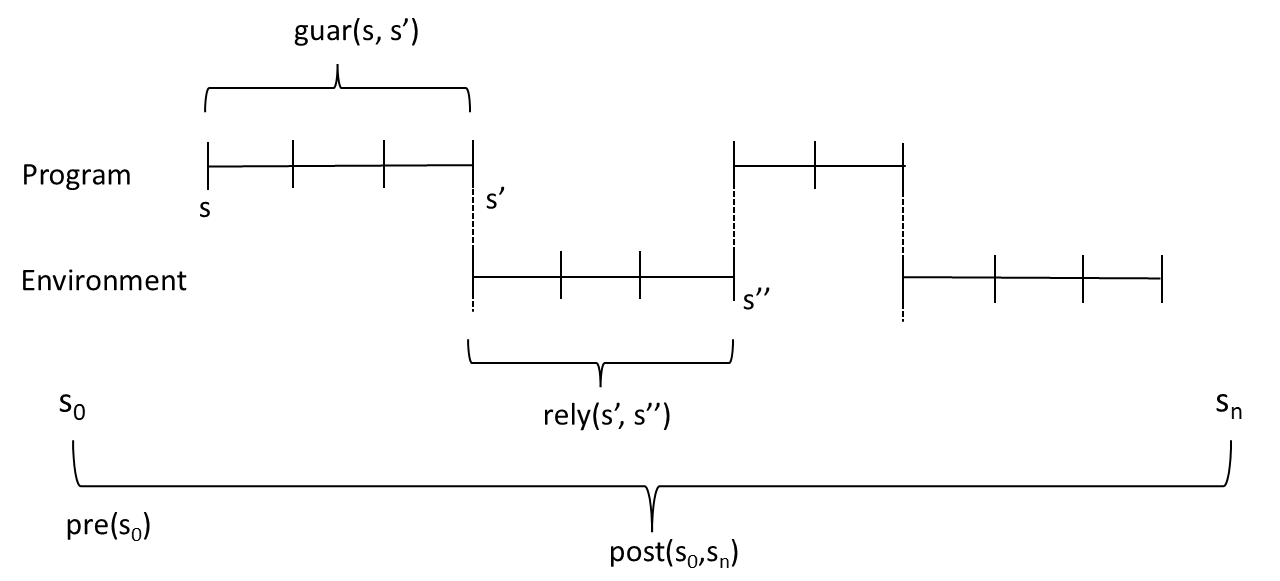}
  \caption{Rely/Guarantee Reasoning (adapted from the version in \cite{Yatapanage25})}  \label{rg}
   \end{center}
\end{figure}

\section{\textit{Holds For All} Reasoning} \label{holdsforall}

For some algorithms, devising suitable rely and guarantee conditions requires the use of shared ghost variables. These are auxiliary variables which are required only for the purpose of proving the verification goals. The use of such variables breaks compositionality, as the development of one component requires an understanding of the details about the internal behaviour of another component. 

The approach suggested in this paper is to instead identify a property that represents what is actually required by a component, and to then reason that the other component will satisfy that property, without referring to the current value of the internal variables of the first component.

Suppose that Component $A$ has an internal variable, $s$. As it is an internal variable, Component $B$ should not be able to determine its current value. Further suppose that Component $A$ requires certain behaviour from its environment, which involves the variable $s$. One option would be to use a ghost variable, allowing the developer of Component $B$ to obtain access to the current value of $s$, for the purpose of verification. Instead, the \textit{Holds For All} rule presented below allows reasoning that Conponent $B$ satisfies the required conditions for any instance of $s$. In other words, Component $B$ ensures that the property will always be satisfied, regardless of the current value of $s$. 

\begin{theorem}{\textit{Holds For All}}
\InfRule{holds-for-all}{guar-B(\sigma,\sigma') \\ p(s,\sigma) \\
 \forall{w}{guar-B(\sigma,\sigma') \, \And \, p(w,\sigma)  \implies q(w,\sigma, \sigma')} \\
p(s,\sigma) \, \And \, q(s,\sigma,\sigma') \implies rely-A(\sigma,\sigma')                  
                     }
                    {rely-A(\sigma,\sigma')}
                    
\end{theorem}  
                    
The idea is to find a suitable property $p$ and a relation $q$ such that if: 
\begin{itemize}
\item the guarantee of one of the components, $B$, holds between the states $\sigma$ and $\sigma'$,
\item a property $p$ holds on $\sigma$ for a specific variable $s$,
\item for any $w$, if $p$ holds on $\sigma$ for $w$ and the guarantee of $B$ holds between $\sigma$ and $\sigma'$, then a relation $q$ holds over the two states and $w$, and
\item  the property $p$ holding on $\sigma$ for the variable $s$, and the relation $q$ holding between $\sigma$ and $\sigma'$ with $s$, implies the rely of the other component, $A$,
\end{itemize} 
then it can be concluded that $rely-A$ holds between $\sigma$ and $\sigma'$.  \\

\noindent
\textbf{Proof.} \\
Since $guar-B$ holds and $p$ holds over $\sigma$ and $s$, then from the third clause, it can be concluded that $q(s,\sigma, \sigma')$ holds. From the last clause, this implies that $rely-A$ holds between $\sigma$ and $\sigma'$.  \qed \vspace{4mm}

The properties $p$ and $q$ are similar to invariants, but are properties that are only required by Component $A$ to hold on environment steps, i.e. steps of Component $B$. Together, they imply the rely of Component $A$. In some cases, it may be possible to simplify the rule so that it only requires the preservation of $p$ over the two states in order to imply the rely of Component $A$.

\section{Ben-Ari Concurrent Garbage Collector}
\label{gcsection} 
It is useful to demonstrate the \textit{Holds For All} reasoning approach on a case study. The Ben-Ari garbage collection algorithm  \cite{BenAri-84} is a mark-and-sweep algorithm where the $Collector$ runs concurrently with the $Mutator$, thus allowing lost memory to be reclaimed while a program is running. Fig. \ref{gc} shows the cycle of behaviour by the two components of the system: the $Mutator$ 
 which uses and loses memory and the $Collector$ which must free the garbage.

   \begin{figure}
   \begin{center}
    \includegraphics[alt={A diagram showing how memory moves from the free set to the busy by the Mutator, is lost by the Mutator, and then reclaimed by the Collector.},width=0.6\linewidth]{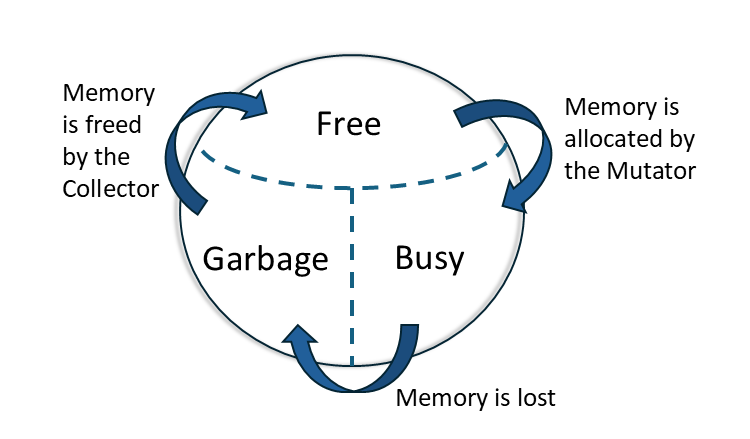}
  \caption{Overview of the Garbage Collector}  \label{gc}
   \end{center}
\end{figure}
 
The $Mutator$ algorithm, which represents an executing program, consists of three actions: redirecting, allocating and removing. A redirect changes a link to point to a different node. Nodes represent memory addresses, and the two terms are used interchangeably in this paper. Nodes can be allocated by taking a new node from the $Free$ set, and links can be removed. The $Mutator$ may choose to perform any of the actions at any time. Note that the removing of a link can be considered to be the same as a redirect, where a link is changed to point to nil.

The algorithm for the $Collector$ operates in three phases: $Mark$, $Sweep$ and $Unmark$. Nodes may be either marked or unmarked. The $Collector$ uses the marking to determine which nodes are garbage and need to be reclaimed. It starts by iterating through all of the nodes in memory and for each marked node, all of the nodes it points to (referred to as the child nodes), are marked as well. After 
every cycle, the number of marked nodes is checked. If it has increased, the cycle runs again, searching
for marked nodes with unmarked children and marking them. An outline of the algorithm for this phase
of the $Collector$'s operation, called \textit{Mark}, is shown in Figure \ref{CollectorAlgorithm}, adapted from \cite{JonesYatapanage19}. 

During the $Sweep$ phase, executed after $Mark$, the $Collector$ reclaims any unmarked nodes and stores them into the $Free$ set. The $Unmark$ phase is responsible for unmarking all of the nodes, except the roots and free set, ready for the next round of $Mark$. This paper focusses on the $Mark$ phase, as interference causes the most impact during this phase.

\begin{figure}

\begin{center}
\begin{minipage}[t]{6cm}
\begin{formula}
Mark \;\DEF\T2
   currCount := \textit{getCount}(marked) \T2
   \kw{do}\T3
       prevCount := currCount \T3      
       Propagate\T3
       currCount := \textit{getCount}(marked) \T2
\kw{while } (currCount != prevCount)
\end{formula}
\end{minipage}
\begin{minipage}[t]{9cm}
\begin{formula}
Propagate \;\DEF\T2
       consid := \emptyset;\T2
       \kw{for each } x \in (Addr \minus consid) \kw{ do}\T3
         \kw{if } x \in marked 
              \kw{ then } \textit{markChildren(x)} \T3
              \kw{else skip} \T3
              \kw{fi};\T3
         consid := consid \union \set{x}\T2
         \kw{end for} 
\end{formula}
\end{minipage}
\end{center}

\caption{Psuedo-code for $Mark$, adapted from \cite{JonesYatapanage19}}\label{CollectorAlgorithm}
\end{figure} 

\subsection{Interference in the Algorithm} \label{inter}
The algorithm as stated above has a significant issue when run concurrently with an active $Mutator$, unless the $Mutator$ performs certain actions. The issue 
is that addresses that are currently in use could be deemed garbage by the $Collector$ and incorrectly freed. To understand why this is the case, consider the following example, shown in Fig. \ref{example}. 
Suppose that the $Collector$ is in the middle of its Mark phase. It has examined address $A$, finding that 
there are no children. The $Mutator$ then interrupts, and changes the pointer from $A$ to point to $C$. The $Mutator$ then removes the link from $B$ to $C$. The $Collector$ then resumes, but it does not check $A$ again. The $Collector$ moves on to $B$, which no longer has any children. Node $D$ is then checked, but is not
marked. The cycle completes, and the number of marked nodes has not increased, so the cycle does not run again. The $Collector$ begins the Sweep phase, where it removes all unmarked nodes. This now includes node
$C$, which is therefore removed, even though it is in use. 

   \begin{figure}
   \begin{center}
    \includegraphics[alt={The steps showing the case where the Collector inadvertently frees active memory. A link is added from A to C and an existing link from B to C is removed. The Collector had already checked A's children and found none, and then moves to B, finding no children. Node C is not identified as active.},width=0.9\linewidth]{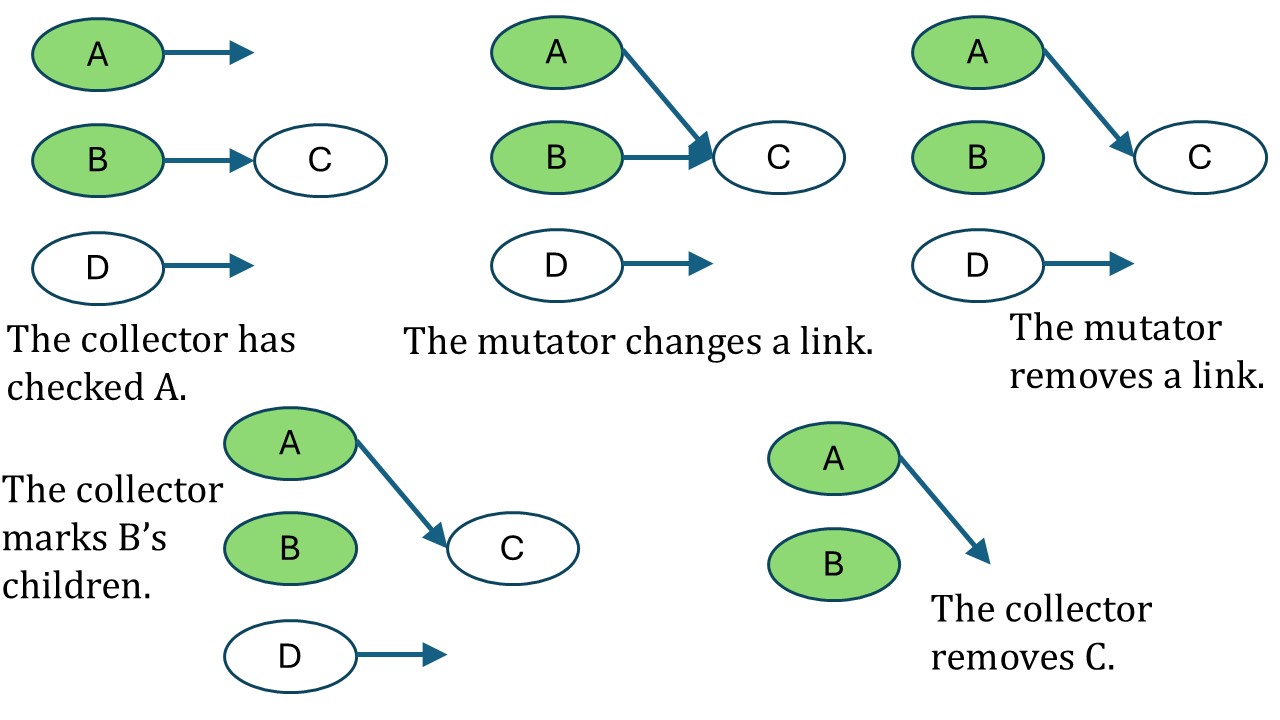}
  \caption{Example of the $Collector$ Incorrectly Removing Addresses (Green nodes are marked)}  \label{example}
   \end{center}
\end{figure}

The Ben-Ari algorithm avoids this situation by requiring the $Mutator$ to perform some actions in order to \textit{help} the $Collector$: whenever the $Mutator$ redirects a pointer in the heap, it must mark the new address. This ensures that the $Collector$ will not inadvertently deallocate memory that is still in use. In the previous example, after the $Mutator$ redirects the pointer from $A$ to point to $C$, it must mark $C$. This ensures that even though the $Collector$ does not mark $C$, at the end of the cycle, the $Collector$ will still recognise that the number of marked nodes has increased, and run the cycle again. On the next cycle, the link from $A$ to $C$ will ensure that $C$ is marked. The essential requirement is that the $Mutator$ must mark the destination node before performing another update, such as a deletion as in the example.

There are many ways of specifying the interference in this scenario. There needs to be a way of specifying the required order of the three actions of the $Mutator$: redirect; mark; delete. The problem only arises when a redirect is followed immediately by a delete. The difficulty with devising a rely condition for the $Collector$ is that it would have to specify this order, whereas rely conditions are binary conditions holding between the start and end states of a block of environment steps. There is no way to refer to steps in the middle. Furthermore, the $Mutator$'s actions might not all fall into a single block of steps. The $Mutator$ might perform a redirect and then be interrupted by the $Collector$. It must be ensured that when the $Mutator$ resumes, it will mark before any further actions. 

One approach for modelling this is to use a ghost variable, as in \cite{JonesYatapanage19}. A ghost variable $tbm$ (or \textit{to be marked}) is used. This is a shared ghost variable, which the $Mutator$ updates whenever there is a pending mark operation. The rely condition of the $Collector$ needs to refer to this ghost variable, requiring the $Mutator$ to mark any pending nodes before any further actions. The disadvantage of this approach is that the $Collector$ must have detailed knowledge about the $Mutator$'s behaviour, i.e. which node has just been modified and is due to be marked.

Another approach would be to use methods such as temporal logic to specify the required order of the steps. When considering only the $Mutator$'s steps on the trace, the temporal logic condition required would state that a redirect step should be followed immediately by a marking step. This would require defining the temporal logic behaviour on projected traces consisting of only the $Mutator$'s steps. However, such an approach requires reasoning over the entire traces, instead of utilising the benefit of rely/guarantee where reasoning is performed over single steps or blocks of steps.

Instead of trying to model the interference directly, it is useful to think of the problem in a different way. The $Collector$ requires the $Mutator$ to mark between changes. However, \textit{why} does it want this? When examining the specific scenario discussed above, it is clear that marking is required to prevent the $Collector$ from missing nodes that are in use. What is not so clear is whether this is the only scenario where nodes can be missed. When verifying by focussing on this specific scenario, it is impossible to be certain that it is the only situation that causes problems. Instead, the method shown in this paper reveals a deeper insight into why marking is needed.

\section{Reasoning About The Collector}
\label{AbstractVariable}

\subsection{Preliminary Definitions}

Before proceeding, some preliminary definitions are required, taken directly from \cite{JonesYatapanage19}, presented using VDM syntax. For further details on VDM, please refer to \cite{Jones90a}. Note that this paper uses the convention of primed variables referring to the \textit{after} state and unprimed variables referring to the \textit{before} state, within rely or guarantee conditions or postconditions.

The state of the concrete representation consists of a set of addresses that form the roots, the heap and a set of free addresses. It is defined as follows: 

\begin{record}{\Sigma\sb{1}}
       roots: \setof{Addr} \\
       hp:  Heap \\ 
       free: \setof{Addr} \\
       marked: \setof{Addr}
\end{record}       
\where
\begin{fn}{inv-\Sigma\sb{1}}{mk-\Sigma\sb{1}(roots, hp, free, marked)} 
%    \dom hp = Addr \And \\
    free \inter reach(roots, hp) = \emptyset %\And \hfill \hbox{upper bound for GC}\\
%   \forall{a \in free}{kids(a, hp) =\emptyset}
 \end{fn}
 \noindent
where the Heap is defined as a mapping from a pair, consisting of an address and an index, to an address: 
\type{Heap}{\mapof{(Addr \x Index)}{Addr}}
\noindent
and the function $reach$ returns all the nodes which are reachable from a given set $s$, by following child links transitively. The invariant ensures that there is no overlap between the free addresses and the ones which are currently in use.

\subsection{Using a Shared Ghost Variable}

In \cite{JonesYatapanage19}, a rely condition for the Collector is given which uses a shared ghost variable known as $tbm$ (\textit{to be marked}), which is a set containing the nodes needing marking after a $Mutator$ change. (This set will always be either empty or a singleton set). The $Collector$'s rely mentions this shared ghost variable. It specifically describes the problematic case discussed in Section \ref{inter}, ensuring that whenever a change occurs, it should be followed by $tbm$ being set to the new destination node, and that no further change can occur until $tbm$ is empty. The relevant clauses of the rely given in \cite{JonesYatapanage19} are as follows: \\

\noindent
           $(\forall*{(a, i) \in \dom{hp}}{hp'(a, i) \neq hp(a, i) \And hp'(a, i) \in Addr \Implies \T8
                 hp'(a, i)  \in marked' \Or tbm' = \set{hp'(a,i)})  \And}\\
           (tbm \neq \emptyset \And tbm' \neq tbm \Implies tbm \subseteq marked' \And tbm' = \emptyset)$ \\
           
This rely condition requires the developer of the $Collector$ to understand exactly how the $Mutator$'s marking algorithm works. The goal of this paper is to avoid such breaks in compositionality. The following section discusses an alternative property which is then used as the basis for applying the \textit{Holds For All} rule.

\subsection{Finding an Alternative Property}

To determine why the $Collector$ requires marking, it is helpful to reason about its behaviour. The $Collector$ attempts to
ensure that all marked nodes in the set of considered nodes, $consid$, have marked children. Without any interference by the $Mutator$, the $Collector$ would 
maintain this property; it is only the $Mutator$'s modification of links while the $Collector$ is still in progress that could undermine this. The
$Mutator$'s marking of all target nodes after a change solves the problem.

These observations suggest that an alternative way of specifying the rely of the $Collector$ is to state that the child of a marked node which has already been considered in the current $Propagate$ round is either marked as
well or can still be marked in the future by $Propagate$. However, such a rely condition would require mentioning the set of nodes that have currently been considered by the $Collector$, i.e. the $consid$ variable. This is a local variable of the $Collector$, which leads to the same issue as using $tbm$; now, the developer of the $Mutator$ has to demonstrate that it maintains this property on the current contents of $consid$, requiring knowledge of the $Collector$'s local variable.

Instead, using the $Holds For All$ rule, the reasoning can be performed by showing that the $Mutator$'s behaviour satisfies the property for all sets. Therefore, in order to make the property as general as possible, instead of using $consid$, the function $can-be-marked$ takes as a parameter an abstract set $s$.

The $can-be-marked$ property states that for every node that is marked and in the set $s$, if one of its children are unmarked, then there must be another path to that node via another node $c$. This node $c$ should be reachable from the roots, or from a marked node that is not in $s$. This ensures that $c$ will eventually be marked on some iteration of $Propagate$.\footnote{It may seem sufficient to state that $c$ is reachable from the roots, but it is possible that a redirection could break the path from the roots to $c$. In such a case, $c$ will still be found in the search via the broken path, if that path was already partially explored.} The $c \notin s \Or c \notin marked$ requirement ensures that $c$ has either not been considered yet by the $Collector$ or has not been marked yet, although it will be marked in the future as it is reachable from a marked node.

\begin{fn}{can-be-marked}{s, mk-\Sigma\sb{1}(roots, hp, free, marked)} \\
  \signature{\setof{Addr} \x \Sigma\sb{1} \to \Bool}
   \forall*{a \in (marked \inter s)}{
          \forall*{i \in \dom hp(a)} {
                hp(a, i) \in (Addr \minus marked) \, \And \, hp(a, i) \in reach(roots,hp) \Implies\\
                \exists*{(c, j) \in \dom{hp}
                } {
                        (c \in reach(roots) \Or c \in reach(marked \inter (Addr \minus s), hp)) \And\\  
                        (c \notin s \Or c \notin marked) \And\\
                        a \neq c \And hp(c,j) = hp(a,i)
                        }}}
\end{fn}

It is necessary to define a property, $all-ok$, which is stronger than $can-be-marked$, stating that there are no marked nodes
with unmarked children. At whatever time the $Collector$ is interrupted, it would have a local $consid$ set containing the addresses that have been explored so far.
At the point that they were added to the $consid$ set, these addresses would all have satisfied $all-ok$, as all marked nodes in the set would have had marked children. 

\begin{fn}{all-ok}{s, mk-\Sigma\sb{1}(roots,hp,free,marked)} \\
  \signature{\setof{Addr} \x \Sigma\sb{1} \to \Bool}
   \forall*{(a,j) \in \dom hp} {a \in s \And a \in marked \And hp(a,j) \in reach(roots,hp) \implies 
   hp(a,j) \in marked}
\end{fn}

%\begin{fn}{is-prob}{a, i, s, mk-\Sigmastart(roots,hp,free,marked)} \\
%  \signature{Addr \x \Nat \x \setof{Addr} \x \Sigmastart \to \Bool}
%a \in s \And a \in marked \And hp(a,j) \in reach(roots,hp) \And\\
% (hp(a,j) \notin s \Or hp(a,j) \notin marked)
%\end{fn}

\vspace{-6mm}

However, there is one situation that changes this. 
If the $Collector$ marks a child of a marked node, it is possible that the child was already in $consid$. In that case, it violates $all-ok$, as there would then be a marked node in $consid$ without necessarily having marked children. Recall that when a node is explored, if it is not marked, it is just added to $consid$ without marking its children. This is not a problem for the $Collector$: as the number of marked nodes increases, the $Propagate$ loop will run again, finding the problem node on the next iteration.

Therefore, instead of the $consid$ set, there is a subset, referred to as $safe$, which are the nodes that the $Collector$ will not necessarily revisit. These are the nodes of $consid$ that still satisfy $all-ok$ at the point when the $Mutator$ interrupts.  

Since the $Collector$ will assume that the $safe$ addresses have finished being examined, the $Mutator$ must ensure that there is some way that will cause the $Collector$ to revisit them. There are two options: either there is another path to the node that is yet to be explored, as defined by $can-be-marked$, or the number of marked nodes have increased, which will cause the $Collector$ to run another cycle of $Propagate$.

This reasoning gives the following rely condition:

%This section presents the full approach of using \textit{Holds For All} reasoning for this problem. For this version, there is no sharing of auxiliary variables between the $Mutator$ and $Collector$. As a result, the $Mutator$'s guarantee and the $Collector$'s rely are different to each other. While the $Mutator$ uses a variable $tbm$ (\textit{to-be-marked}), to indicate which node is still to be marked after a $Mutator$ operation, as in \cite{JonesYatapanage19}, unlike that version, it is a local auxiliary variable and not a global shared variable, i.e. the $Mutator$ is the only component that has access to it. Similarly, the $Collector$ makes use of the $consid$ set, which is its local auxiliary variable and cannot be accessed by the $Mutator$. The $consid$ set represents the set of variables which the Collector has already considered in the current round of marking. Marked nodes are added to $consid$ only after all of their children have been marked. 
   
%   The rely of the $Collector$ is that if the initial state satisfied $can-be-marked$ using the $consid$ set, then the final state also satisfies $can-be-marked$. Additionally, the marked nodes should only increase, to prevent the $Mutator$ from unmarking nodes. If a change has been made, the new destination node should be reachable, to prevent connections being made to garbage nodes that were lost.

\begin{fn}{rely-Collector}{\sigma, \sigma'}\\
 \signature{\Sigma\sb{1} \x \Sigma\sb{1} \to \Bool}
\sigma.marked \subseteq \sigma'.marked \And \\
  (\forall*{(c, i) \in \dom \sigma.hp}{
   \sigma'.hp(c, i) \neq \sigma.hp(c, i) \And
    \sigma'.hp(c, i) \in (Addr \minus \sigma.free) \Implies
      \sigma'.hp(c, i) \in reach(\sigma.roots, \sigma.hp) )\And}  \\
     (all-ok(\sigma.safe,\sigma) \Implies can-be-marked(\sigma.safe,\sigma') \Or (\sigma.marked \subset \sigma'.marked))
\end{fn}

%\begin{fn}{mark-change}{s, mk-\Sigmastart(roots,hp,f,marked), mk-\Sigmastart(r',hp',f',marked')} \\
%  \signature{ \setof{Addr} \x \Sigmastart \x \Sigmastart \to \Bool}
% \forall*{(a,j)\in \dom hp}{(a \in s \And a \in marked \And hp(a,j) \notin marked 
%\Implies\T2
% (hp(a,j) \in marked' \Or \T2
%    (\exists{(c,i)\in \dom hp}{hp'(c,i) = hp(a,j) \And c \in reach(roots,hp') \T3
%     \And a \neq c \And (c \notin s \Or c \notin marked))))}}
%\end{fn}

The first clause of the rely condition states that the marked nodes should only increase. The second clause ensures that new links cannot be made with lost nodes. The final clause states that if there is a set of nodes which satisfied $all-ok$ at $\sigma$, then either that set should satisfy $can-be-marked$ at $\sigma'$ or the number of marked nodes should have increased. To simplify the presentation, $safe$ is referred to as if it is part of the state, although it is actually a local variable of the $Collector$.
  
\section{Reasoning Using \textit{Holds For All}} \label{holdsreasoning}
   
An interesting observation is that the $Mutator$ intuitively satisfies the conditions required by the $Collector$ in $rely-Collector$
without needing to understand the notion of $safe$, nor even the contents of $safe$; these conditions are simply satisfied by the $Mutator$'s actions, in particular
its marking of a target node before proceeding to make another change. This suggests that the $Mutator$'s actions satisfy these conditions
not just on the $safe$ set, but on any similar set that possesses certain qualities. This abstract set forms the basis of this paper's approach for reasoning about this problem. 

For this particular application of the \textit{Holds For All} rule, the variable $s$ is the $safe$ set. This allows reasoning about the $safe$ set of the $Collector$ without actually referring to it. The $Mutator$ is required to show that the third clause of the \textit{Holds For All} rule holds over any set of addresses, $w$. This proof does not require the $Mutator$ to have knowledge of what addresses are currently in the $safe$ set. Using the above reasoning, this establishes that the rely condition of the $Collector$ holds, which refers to the $safe$ set.

%\begin{definition}\textit{Holds-for-all}
%
%
%$\forall {v, \sigma, \sigma'}$
%\begin{formula}
%\left(
%\left(
%\begin{array}{l}  
%                p(v,\sigma) \, \And \\
%                ( \forall {\mu, \mu', x} {p(x, \mu) \And guar-B(\mu,\mu') \implies p(x,\mu'))} \, \And \\
%                guar-B(\sigma,\sigma') \, \And\\
%                ( \forall{\mu, \mu', x}{p(x, \mu) \And p(x, \mu') \implies rely-A(\mu, \mu')})
%\end{array}
%\right)                  \Implies
%                 rely-A(\sigma,\sigma')  
%                 \right)      
%\end{formula}     
            
%where $v$ is a variable and $\sigma$ and $\sigma'$ are states. 
%\end{definition}

\subsection{The $Mutator$ Proof Obligations}

Applying the \textit{Holds For All} theorem, with $all-ok(s,\sigma')$ as $p$ and $can-be-marked(s,\sigma') \Or (\sigma'.marked > \sigma.marked)$ as $q$, the proof obligation of the $Mutator$ is to show that for any set of addresses, if $p$ holds on $\sigma$, then $q$ holds between $\sigma$ and $\sigma'$.

The $Mutator$ consists of two atomic operations, $Change$ and $MarkNode$, which are
 described  below. As
these operations are atomic, the $Mutator$'s overall guarantee is formed by the postconditions of each of the operations.

The $post-Change$ condition ensures that if a redirection occurred, the new target node is stored as $needsMarking$ and that node was reachable from the roots in the original $hp$. Note that it may no longer be reachable in $hp'$, but as long as it was reachable in the previous state, it satisfies the required conditions of the $Collector$.

Note that, to make the following definitions easier to read, $needsMarking$ has been used as if it is part of the state, although it is actually a local variable of the $Mutator$. As these are all part of the $Mutator$'s obligations, it should be accessible. 

\begin{fn}{post-Change}{mk-\Sigma\sb{1}(roots, hp, free, marked, needsMarking), \R mk-\Sigma\sb{1}(roots', hp', free', marked', needsMarking')}\\
 \signature{\Sigma\sb{1} \x \Sigma\sb{1} \to \Bool}
 \Let (x,k) = \uniqueval{(x,k)} {hp(x,k) \neq hp'(x,k)} \In 
        ((hp'(x,k) = \nil \And needsMarking' = \nil) \Or \\            
         (needsMarking'= hp'(x,k) \And
         hp'(x,k) \in reach(roots,hp))) \\
       \And marked' = marked 
\end{fn}

\noindent
After $MarkNode$ finishes, the node that was stored as $node-to-mark$ must have been marked. 

\begin{fn}{post-MarkNode}{mk-\Sigma\sb{1}(roots, hp, free, marked, needsMarking), \R
mk-\Sigma\sb{1}(roots', hp', free', marked', needsMarking')'}\\
 \signature{\Sigma\sb{1} \x \Sigma\sb{1} \to \Bool}
((needsMarking = \nil \, \And \, marked' = marked) \Or\\
marked' = marked \, \union \, \set{needsMarking}) \, \And \, needsMarking' = \nil
\end{fn}

Starting from an $all-ok$ state, the $Mutator$ can execute $post-Change$. This results in $can-be-marked$ holding on the final state, given by Lemma \ref{canbemarkedRule}.
% I changed mut-prop4-b to c and vice versa - watch out whether this affects anything later!
\begin{lemma} \label{canbemarkedRule}
$all-ok(s,\sigma) \, \And \,  post-Change(\sigma, \sigma') \T9
                \implies can-be-marked(s,\sigma') \, \And \, onlyProblem(s,\sigma')$  
\end{lemma}
The above lemma also results in the condition $onlyProblem$. The $onlyProblem$ property ensures that $needsMarking$ is the
only unmarked node with a marked parent in the set. The need for this new property is that while
 $can-be-marked$ ensures that all nodes requiring marking are accessible via another unexplored node, it does not
  ensure that such nodes are recorded by the $needsMarking$ variable. This is needed for specifying the requirements under which $MarkNode$ executes. 
  
\begin{fn}{onlyProblem}{s, mk-\Sigma\sb{1}(roots,hp,free,marked,needsMarking)} \\
  \signature{\setof{Addr} \x \Sigma\sb{1} \to \Bool}
\forall*{(a,j) \in \dom hp}{(needsMarking \neq hp(a,j) \And a \in s \And a \in marked \Implies \T4
        hp(a,j) = \nil \Or hp(a,j) \in marked)}
\end{fn}    

If $can-be-marked$ holds on a set $s$, as well as the property $onlyProblem$, then executing $MarkNode$ will ensure
that the set either satisfies $all-ok$ or the number of marked nodes has increased. This is given by Lemma \ref{markRule}. If the $Mutator$ was responsible for adding an unmarked node, $onlyProblem$ holds. Therefore, after the $MarkNode$ operation, in
which $needsMarking$ is marked, ideally, there would not be any more problem nodes. However, it is possible for the $Mutator$ to mark a destination node that is already in the set $s$, similarly to how the $Collector$ could mark a child that is already in $consid$.
In such a case, $all-ok$ would not hold, but the number of marked nodes would have increased.\footnote{Note that it is possible for a $MarkNode$ operation to not increase the number of marked nodes, if the node it attempts to mark is already marked. Nevertheless, since the node is already marked, such cases cannot result in a new problem node, and therefore, $can-be-marked$ still holds.}

  \begin{lemma} \label{markRule}                  
$can-be-marked(s,\sigma) \, \And \, onlyProblem(s,\sigma) \, \And \,
                     post-MarkNode(\sigma, \sigma') \implies \R
                    all-ok(s,\sigma') \Or (\sigma.marked \subset \sigma'.marked)$
                    \end{lemma}

Lemma \ref{canbemarkedRule} requires $all-ok$ before $Change$ executes. Since the Ben-Ari algorithm requires the $Mutator$ to mark in between every change, there needs to be a way of constructing a suitable set that satisfies $all-ok$ after a $MarkNode$ operation. 
After marking, as discussed above, there may be at most one new problem node. Thus, the set of all nodes in $s$ except that node satisfies $all-ok$. 

The alternation of $Change$ and $MarkNode$ potentially results in a smaller set each time that $MarkNode$ executes. Each time the set reduces, the number of marked nodes has increased, by Lemma \ref{markRule}. Therefore, for any sequence of $Mutator$ steps, either $can-be-marked$ can be shown to hold at the end of the sequence, or the number of marked nodes would have increased. This satisfies the requirements of $rely-Collector$.

\subsection{Establishing the Rely of the Collector} \label{S-CollNoShare} 

In order to use the $HoldsForAll$ rule, it now remains to be shown that using $all-ok(s,\sigma')$ as $p$ and $can-be-marked(s,\sigma') \Or (\sigma.marked \subset \sigma'.marked)$ as $q$ implies
$rely-Collector$. This is trivial for the last clause of $rely-Collector$. The other clauses do not involve any internal variables of the $Collector$, so it is simple to directly prove that 
the $Mutator$ satisfies these clauses.

As explained earlier, the set which satisfies $all-ok$ or $can-be-marked$ may decrease after steps of the $Collector$ or $Mutator$. The $safe$ set represents the set for which the $Collector$ has preserved $all-ok/can-be-marked$. If the $Mutator$ reduces the set, it must increase the number of marked nodes, as shown by the $q$ clause. This guarantees that another round of marking will take place and the required nodes will be identified.

\section{Discussion} \label{Discussion}

The $can-be-marked$ property reveals what is actually required by the $Collector$ in order to ensure that the marking phase identifies all active addresses. This is significant because identifying the core requirements aids in developing appropriate algorithms. While in this case, the Ben-Ari algorithm was an existing algorithm, it is useful to understand what is required of the $Mutator$. Importantly, it demonstrates that the $Mutator$ does not have to mark after every change, if it were possible to determine whether a change affects an address that was already considered. There does not seem to be any obvious way of determining this, so marking after every change overcompensates. Therefore, while identifying the $can-be-marked$ property does not reveal an alternative algorithm, it clearly identifies what the $Collector$ requires in order to function correctly. For other algorithms, identifying the core requirements of a component may lead to more efficient algorithms.

In terms of compositionality, it could be argued that the approach used in this paper still requires the developers of each component to learn about the other component. In particular, the reasoning about the $Mutator$ requires showing that it preserves $can-be-marked$ or $all-ok$, thus requiring knowledge of these properties, which only exist because of the $Collector$'s requirements. Nevertheless, there is a significant difference in the compositionality of the proofs. Requiring a component to have knowledge of a property required by another is in line with the standard use of rely/guarantee: each rely condition specifies what a component requires of its environment, thereby requiring the environment processes to satisfy those requirements. Using shared ghost variables instead requires reasoning about the exact value of the local variables of other components.

Consider the version of the $Collector$'s rely using the shared ghost variable $tbm$. The use of $tbm$ in the rely requires the developer of the $Collector$ to have full knowledge about what the $tbm$ variable represents and how the $Mutator$ uses it. The proof obligations include showing that the $Collector$ can reach its postcondition. The reasoning would have to use $tbm$ even for this proof obligation, which normally should not involve such details about another component. More importantly, the algorithm of the $Collector$ does not have anything to do with $tbm$, so this reasoning is not aligned with the algorithm. Suppose that an alternative algorithm for the $Mutator$ were proposed. The verification proofs for the $Collector$ would all have to be revised, as they used the $tbm$ variable and the particular use of $tbm$ of the current $Mutator$ algorithm.

On the other hand, the remaining proof obligations of the approach in this paper, which includes showing that the $Mutator$ can reach its postcondition, do not require any knowledge about the $consid$ variable of the $Collector$. The \textit{Holds For All} approach is only necessary to show that the $Mutator$'s guarantee achieves what is required by the $Collector$'s rely. It does not affect the rest of the proof obligations for the $Mutator$. This aligns with the algorithm, in which the $Mutator$ does not do anything with $consid$. If the $Collector$ algorithm were replaced, it would not affect the proof obligations for the $Mutator$, except the one showing that its guarantee satisfies the rely of the $Collector$. This is what should be expected; changing one component should only affect the proof obligations involving their interference, not the proofs that another component reaches its postcondition.

Instead of the \textit{Holds For All} approach, it is also possible to use $consid$ as a shared ghost variable. This would not simplify the proofs, because in order to reason that the $Collector$'s rely is satisfied, even with knowledge of the current value of $consid$, the reasoning would still have to show that the property is satisfied for any such set. Furthermore, if the guarantee of the $Mutator$ was written to include mentioning $consid$ directly, then the proof obligations of the $Mutator$, such as showing that it reaches its postcondition, would be affected.

\subsection{Related Work} \label{Related}
Havelund \cite{havelund1999mechanical} verifies the Ben-Ari algorithm using the PVS theorem prover and also with the Mur$\phi$ model checker in order to perform a comparison of the two techniques. A requirement is formalised that no marked node can point to an unmarked node. The problem is approached by specifying invariants and reasoning over the system as a whole, as the goal was not compositionality.

Another attempt at verifying Ben-Ari's algorithm is \cite{NietoEsparza00}, which uses Owicki-Gries reasoning \cite{OwickiGries76}. The proof uses two auxiliary variables. A few other attempts at the verification of Ben-Ari's algorithm exist, but some were shown to contain flaws, such as Ben-Ari's own proof of correctness. Havelund \cite{havelund1999mechanical} gives an interesting account of several authors who encountered similar mistakes in reasoning about Ben-Ari's algorithm, demonstrating how difficult it is to prove its correctness, due to the intricate interactions between the components. For this reason, identifying a clear requirement for the $Collector$ to perform correctly is essential.

The algorithm verified in \cite{Zakowski2017} allows the two components to have access to each other's current state. If the algorithm itself is designed that way, it makes sense to verify it in a similar way. The approach shown in this paper avoids reasoning with shared variables when the algorithm being verified does not contain such sharing.

\section{Conclusion} \label{Conclusion}
This paper has proposed an approach for verifying algorithms with two components, known as \textit{Holds For All} reasoning, for reasoning about the local variables of another component, without the need for shared auxiliary variables. The idea is to identify a general property that holds over an abstract, universally quantified variable, which is then instantiated to the particular local variable required. The proofs are currently being formalised in the theorem prover Isabelle.

Apart from the benefit of avoiding breaks in compositionality, the approach is useful because it helps to identify what the required core property of the component is. Even though the rely condition of the $Collector$ focussed on the $safe$ local variable, the actual needs of the $Collector$ were identified through the $can-be-marked$ property. The identification of the core property also provides a deeper understanding of how Ben-Ari's algorithm works, and in particular, \textit{why} it works. The same approach could be useful for understanding further classic algorithms, for which it is difficult to understand why they work. For future work, it would be interesting to study further examples to see what other approaches it could be useful for.
\\
\section{Acknowledgements} 
The author would like to thank Roger Su for helping with the Isabelle mechanisation of the proofs. GenAI is being used for the Isabelle proof generation, but was not used for the hand-written proofs supporting the claims of this paper. The author would also like to thank Cliff Jones for discussions on an early version of the ideas and the anonymous reviewers for their useful suggestions.

\bibliography{parallel}

@string{acta = {Acta Informatica}}

@string{TOPLAS = "ACM Trans.\ Prog.\ Lang.\ and Sys."}

@incollection{HayesJones18,
	Author = {Hayes, I. J. and Jones, C. B.},
	oAuthor = {Ian J. Hayes and Cliff B. Jones},
	BookTitle =  "Engineering Trustworthy Software Systems",
	oBookTitle = "Engineering Trustworthy Software Systems: Third International School, SETSS 2017 Chongqing, China, April 17--22, 2017",
	Editor = {J. P. Bowen and Z. Liu and Z. Zhang},
	oEditor = {Bowen, Jonathan P. and Liu, Zhiming and Zhang, Zili},
	Publisher = "Springer International Publishing",
	Address = "Cham",
	Series = {LNCS},
	Volume = {11174},
	Title = {A Guide to Rely/Guarantee Thinking},
	Pages = {1--38},
	Issn = {0302-9743},
	Isbn = "978-3-030-02927-2",
	Doi = {10.1007/978-3-030-02928-9_1},
	Year = {2018}}

@book{Jones90a,
  author = {C. B. Jones},
  edition = {Second},
  isbn = {0-13-880733-7},
  length = {333 pages},
  publisher = {Prentice Hall International},
  title = {Systematic Software Development using VDM},
  url = {http://homepages.cs.ncl.ac.uk/cliff.jones/ftp-stuff/Jones1990.pdf},
  year = {1990}
}

@article{OwickiGries76,
  author = {S.~S. Owicki and D. Gries},
  journal = acta,
  issn={0001-5903},
  pages = {319--340},
  title = {An axiomatic proof technique for parallel programs {I}},
  volume = {6},
  number={4},
  doi={10.1007/BF00268134},
  ourl={http://dx.doi.org/10.1007/BF00268134},
  publisher={Springer-Verlag},
  language={English},
  year = {1976}
}

@inproceedings{Jones83a,
  author = {C. B. Jones},
  booktitle = {Proceedings of IFIP'83},
  pages = {321--332},
  publisher = {North-Holland},
  title = {Specification and Design of (Parallel) Programs},
  year = 1983
}

@article{Jones83b,
  author = {C. B. Jones},
  doi = {10.1145/69575.69577},
  journal = {ACM ToPLaS},
  number = {4},
  pages = {596--619},
  publisher = "ACM",
  title = {Tentative Steps Toward a Development Method for Interfering Programs},
  ourl = {http://doi.acm.org/10.1145/69575.69577},
  url-2 = {http://dx.doi.org/10.1145/69575.69577},
  volume = {5},
  year = 1983
}

@incollection{marriage-RGSep,
year={2007},
booktitle={CONCUR 2007 -- Concurrency Theory},
volume={4703},
series={Lecture Notes in Computer Science},
editor={Caires, Lu{\'i}s and Vasconcelos, Vascot},
title={A Marriage of Rely/Guarantee and Separation Logic},
url={http://dx.doi.org/10.1007/978-3-540-74407-8_18},
publisher={Springer},
author={Vafeiadis, Viktor and Parkinson, Matthew},
pages={256-271}
}

@article{JonesYatapanage19,
	Author = {Jones, Cliff B. and Yatapanage, Nisansala},
  title = "Investigating the limits of rely/guarantee
relations based on a concurrent garbage
collector example",
  journal = "Formal Aspects of Computing",
  volume = {31},
  Pages = {353--374},
  year = "2019",
}

@article{PVEaCfC,
  title = "Possible values: Exploring a concept for concurrency",
  journal = "Journal of Logical and Algebraic Methods in Programming",
  volume = "85",
  number = "5, Part 2",
  pages = "972--984",
  month = "August",
  onote = "Accepted 8 January 2016. Available online 13 January 2016",
  xnote = "Articles dedicated to Prof. J. N. Oliveira on the occasion of his 60th birthday",
  issn = "2352-2208",
  doi = "http://dx.doi.org/10.1016/j.jlamp.2016.01.002",
  url = "http://www.sciencedirect.com/science/article/pii/S2352220816000031",
  author = "Cliff B. Jones and Ian J. Hayes",
  year = "2016"
}

@Article{BenAri-84, 
  author = 	 {Mordechai Ben-Ari},
  title = 	 {Algorithms for on-the-fly garbage collection},
  journal = 	 {ACM Transactions on Programming Languages and Systems},
  year = 	 {1984},
  volume = 	 {6},
  number = 	 {3},
  pages = 	 {333--344}
}

@inproceedings{NietoEsparza00,
  author = {Leonor Prensa Nieto and Javier Esparza},
  title     = {Verifying Single and Multi-mutator Garbage Collectors with {O}wicki-{G}ries in {I}sabelle/{HOL}},
  booktitle = {MFCS 2000},
  series    = {LNCS},
  volume    = {1893},
  pages     = {619--628},
  publisher = {Springer},
  year      = {2000}
}

@inproceedings{Zakowski2017,
  AUTHOR = {Yannick Zakowski and
               David Cachera and
               Delphine Demange and
               Gustavo Petri and
               David Pichardie and
               Suresh Jagannathan and
               Jan Vitek},
  TITLE = {Verifying a Concurrent Garbage Collector Using a Rely-Guarantee Methodology},
  BOOKTITLE = {Interactive Theorem Proving - 8th International Conference, {ITP} 2017, Bras{\'{\i}}lia, Brazil, September 26-29, 2017, Proceedings},
  YEAR = {2017},
  EDITOR = {Mauricio Ayala{-}Rinc{\'{o}}n and
               C{\'{e}}sar A. Mu{\~{n}}oz},
  VOLUME = {10499},
  SERIES = {Lecture Notes in Computer Science},
  PAGES = {496--513},
  PUBLISHER = {Springer-Verlag}
}

@inproceedings{havelund1999mechanical,
  title={Mechanical verification of a garbage collector},
  author={Havelund, Klaus},
  booktitle={International Parallel Processing Symposium},
  pages={1258--1283},
  year={1999},
  organization={Springer}
}

@incollection{Yatapanage25,
  author       = {Nisansala P. Yatapanage},
  editor       = {Ana Cavalcanti and
                  James Baxter},
  title        = {Exploring the Boundaries of Rely/Guarantee and Links to Linearisability},
  booktitle    = {The Practice of Formal Methods: Essays in Honour of Cliff Jones, Part
                  {II}},
  series       = {Lecture Notes in Computer Science},
  volume       = {14781},
  pages        = {247--267},
  publisher    = {Springer},
  year         = {2024},
  url          = {https://doi.org/10.1007/978-3-031-66673-5\_13},
  doi          = {10.1007/978-3-031-66673-5\_13},
  bibsource    = {dblp computer science bibliography, https://dblp.org}
}

\newpage

\section{Appendix}
This Appendix contains the proof of Lemma 1. The other lemmas are simple to prove and less interesting to discuss in detail.

Proof of Lemma \ref{canbemarkedRule}: \\
Assume $all-ok(s,\sigma)$ and $post-change(\sigma,\sigma')$. \\

From $post-Change(\sigma,\sigma')$, let $(x,k) = \uniqueval{(x,k)} {hp(x,k) \neq hp'(x,k)}$ (the change which occurred). There are two cases from $post-change(\sigma,\sigma')$ :\\

Case 1: $hp'(x,k) = \nil \And tbm' = \nil$: From the assumption of $all-ok(s,\sigma)$, there were no marked nodes in $s$ with unmarked children. Therefore, removing a link will not change this and so, $all-ok(s,\sigma')$ holds, which implies $can-be-marked(s,\sigma')$. \\

Case 2: $tbm' = hp'(x,k) \And hp'(x,k) \in reach(roots,hp)$: \\
$\exists{(c,j) \in dom(hp)}{hp(c,j) = hp'(x,k)}$  (There had to be another node $c$ pointing to the new child of $a$), from $hp'(x,k) \in reach(roots,hp)$ in $post-change(\sigma,\sigma')$. \\

Either $c \notin reach(roots,hp')$ or $c \in reach(roots,hp')$

Case 2.1: $c \notin reach(roots,hp')$:  \\
As only one link was changed, this must mean that $c$ was only reachable via the original $hp(x,k)$, so the link is broken: $c \in reach(\set{hp(x,k)},hp)$ \\

Case 2.1.1: $x \in s \And x \in marked$: \\
Therefore, $hp(x,k) \in marked$, from $all-ok(s,\sigma)$. \\
Let $q$ be the address such that $q \in (reach(\set{hp(x,k},\sigma) \union \set{hp(x,k)}))$ and \\
$q \in marked \And \exists{r \in dom(hp(q)}{hp(q,r) \notin marked \And 
c \in (reach(\set{r},\sigma) \union \set{r}}) $ \\
($q$ is the first marked node with an unmarked child that is an ancestor of $c$). \\
Therefore, $c$ satisfies the conditions of $can-be-marked$, to provide an alternative path to $hp'(x,k)$. \\
From $all-ok(s,\sigma)$, there is no node $y \in (marked \inter s)$ such that $\exists{j \in hp(y)}{hp(y,j) \notin marked}$. Since $marked' \geq marked$, $\forall{z \in (marked' \inter s)}{\forall{h \in hp'(z)}{hp'(z,h) \in marked}}$. \\
Therefore, $can-be-marked(s,\sigma)$ holds. \\

Case 2.1.2: $x \notin s \Or x \notin marked$: \\
From $all-ok(s,\sigma)$ and the fact that the only change was $hp(x,k)$ redirecting to $hp'(x,k)$, $can-be-marked$ holds vacuously. \\

Case 2.2: $c \in reach(roots,hp')$:  \\
Since $roots \subseteq marked$ in all states, this gives that $c \in reach(marked,hp')$. \\

%Either $c \in reach(s,hp')$ or $c \notin reach(s,hp')$. \\
%
%Case 2.2.1: $c \in reach(s,hp')$: \\

Case 2.2.1: $c \in (s \union marked)$: \\
Then, from $all-ok(s,\sigma)$, $hp(x,k) \in marked$. \\
From $all-ok(s,\sigma)$ and the fact that the only change was $hp(x,k)$ redirecting to $hp'(x,k)$, $can-be-marked$ holds vacuously. \\

Case 2.2.2: $(c \notin s) \Or (c \notin marked)$: \\
Therefore, $c$ satisfies the conditions of $can-be-marked$, to provide an alternative path to $hp'(x,k)$. \\

Showing that $tbm-put$ holds is straight-forward using $all-ok(s,\sigma)$ and $post-Change(\sigma,\sigma')$, as $post-Change$ sets $tbm$ to the new destination node, and from 
$all-ok(s,\sigma)$, there are no other nodes in $s$ and $marked$ with unmarked children.

%\input{S-conf-1}

%%%%%%%%%%%%%%%%%%%%%%%%%%%%%%%%%%%%%%%%%%%%%%%%%%%%%%
\end{document}